\documentclass[a4paper,11pt]{article}
\usepackage{amsthm}
\usepackage{begin}
\usepackage{algorithm2e}
\usepackage{color}
\usepackage{graphicx}
\usepackage{makeidx}
\usepackage{amsmath}
\usepackage{amssymb}
\usepackage{float}

\titlespacing{\section}{0pt}{3.5ex}{1.5ex}
\titlespacing{\subsection}{0pt}{2.5ex}{0.75ex}
\titlespacing{\subsubsection}{0pt}{2.5ex}{0.75ex}
\titlespacing{\theorem}{0pt}{4.5ex}{0.75ex}

\begin{document}

\newtheorem{theorem}{Theorem}
\newtheorem{lemma}{Lemma}
\newtheorem{definition}{Definition}

\title{On Byzantine Broadcast in Planar Graphs}
\author{Alexandre Maurer$^1$ and S\'{e}bastien Tixeuil$^{1,2}$\\
$^1$ UPMC Sorbonne Universit\'{e}s, Paris, France\\
$^2$ Institut Universitaire de France\\
E-mail: Alexandre.Maurer@lip6.fr, Sebastien.Tixeuil@lip6.fr\\
Phone: +33 1 44 27 87 62, +33 1 44 27 87 75
}

\maketitle

\begin{abstract}
We consider the problem of reliably broadcasting information in a multihop asynchronous network in the presence of Byzantine failures: some nodes may exhibit unpredictable malicious behavior. We focus on completely decentralized solutions. Few Byzantine-robust algorithms exist for loosely connected networks. A recent solution guarantees reliable broadcast on a torus when $D > 4$, $D$ being the minimal distance between two Byzantine nodes.

In this paper, we generalize this result to $4$-connected planar graphs. We show that reliable broadcast can be guaranteed when $D > Z$, $Z$ being the maximal number of edges per polygon. We also show that this bound on $D$ is a lower bound for this class of graphs. Our solution has the same time complexity as a simple broadcast. This is also the first solution where the memory required increases linearly (instead of exponentially) with the size of transmitted information.
\end{abstract}

\vspace{5mm}

\paragraph{Important disclaimer}  These results have NOT yet been published in an international conference or journal. This is just a technical report presenting intermediary and incomplete results. A generalized version of these results may be under submission.

\vspace{7mm}

\section{Introduction}

As modern networks grow larger, they become more likely to fail, as nodes may be subject to crashes, attacks, transient bit flips, etc. To encompass all possible cases, we consider the most general model of failure: the \emph{Byzantine} model \cite{LSP82j}, where the failing nodes can exhibit arbitrary malicious behavior. In other words, tolerating Byzantine nodes implies guaranteeing they are not able to cause problems in the correct part of the network.

In this paper, we study the problem of reliably broadcasting information in a multihop network. In the ideal case, the source node sends the information to its neighbors, that in turn send it to their own neighbors, and so forth (this is denoted in the sequel as a ``simple broadcast''). However, a single Byzantine node can forward a false information and lie to the entire network. Our goal is to design a solution that guarantees reliable broadcast in the presence of Byzantine retransmitters.

\paragraph{Related works.}

Many Byzantine-robust protocols are based on \emph{cryptography} \cite{CL99c,DFS05c}: the nodes use digital signatures to authenticate the sender across multiple hops. However, as the malicious nodes are supposed to ignore some cryptographic secrets, their behavior cannot be considered as \emph{entirely} arbitrary.
Besides, manipulating asymmetric cryptography requires important resources, which may not always be available.
The most important point is that cryptography requires some degree of trusted infrastructure to initially distributes public and private keys: therefore, if this initial infrastructure fails, the whole network fails. Yet, we want to design a totally decentralized solution, where any element can fail independently without compromising the whole system.
For these reasons, we focus on non-cryptographic solutions.

Cryptography-free solutions have first been studied in completely connected networks~\cite{LSP82j,AW98b,MMR03j,MRRS01c,MS03j}: a node can directly communicate with any other node, which implies the presence of a channel between each pair of nodes. Therefore, these approaches are hardly scalable, as the number of channels per node can be physically limited. We thus study solutions in multihop networks, where a node must rely on other nodes to broadcast informations.

A notable class of algorithms tolerates Byzantine failures with either space~\cite{MT07j,NA02c,SOM05c} or time~\cite{MT06cb,DMT11cb,DMT11j,DMT10cd,DMT10ca} locality. Space local algorithms try to contain the fault as close to its source as possible. This is only applicable to the problems where the information from distant nodes is unimportant: vertex coloring, link coloring, dining philosophers, etc. Also, time local algorithms presented so far can hold at most one Byzantine node, and are not able to mask the effect of Byzantine actions. Thus, this approach is not applicable to reliable broadcast.

In \cite{D82j}, it was shown that, for agreement in the presence of up to $k$ Byzantine nodes, it is necessary and sufficient that the network is $(2k+1)$-connected, and that the number of nodes in the system is at least $3k+1$. However, this solution assumes that the topology is known to every node, and that the network is synchronous. Both requirements have been relaxed in \cite{NT09j}: the topology is unknown and the scheduling is asynchronous. Yet, this solution retains $2k+1$ connectivity for reliable broadcast and $k+1$ connectivity for failure detection.

Another existing approach is based on the fraction of Byzantine neighbors per node. Solutions have been proposed for nodes organized on a lattice \cite{K04c,BV05c}. Reliable broadcast was shown possible if every node has strictly less than a $1/4$ fraction of Byzantine neighbors. This result was later generalized to other topologies \cite{PP05j}, assuming that each node knows the global topology.

All aforementioned approaches are hardly applicable to loosely connected networks, where each node has a limited (possibly upper bounded by a constant) number of neighbors. For instance, on a torus topology (see Figure~\ref{figgraph}), no existing solution can tolerate more than one Byzantine node. Efficient solutions have been proposed for such networks \cite{CtrZ,Scalbyz}, but only give probabilistic guarantees, and require the nodes to know their position in the network. This last requirement was relaxed in \cite{Trig}: reliable broadcast is guaranteed on a torus when $D > 4$, $D$ being the minimal number of hops between two Byzantine nodes.

\begin{figure}[H]
\begin{center}
\includegraphics[width=13cm]{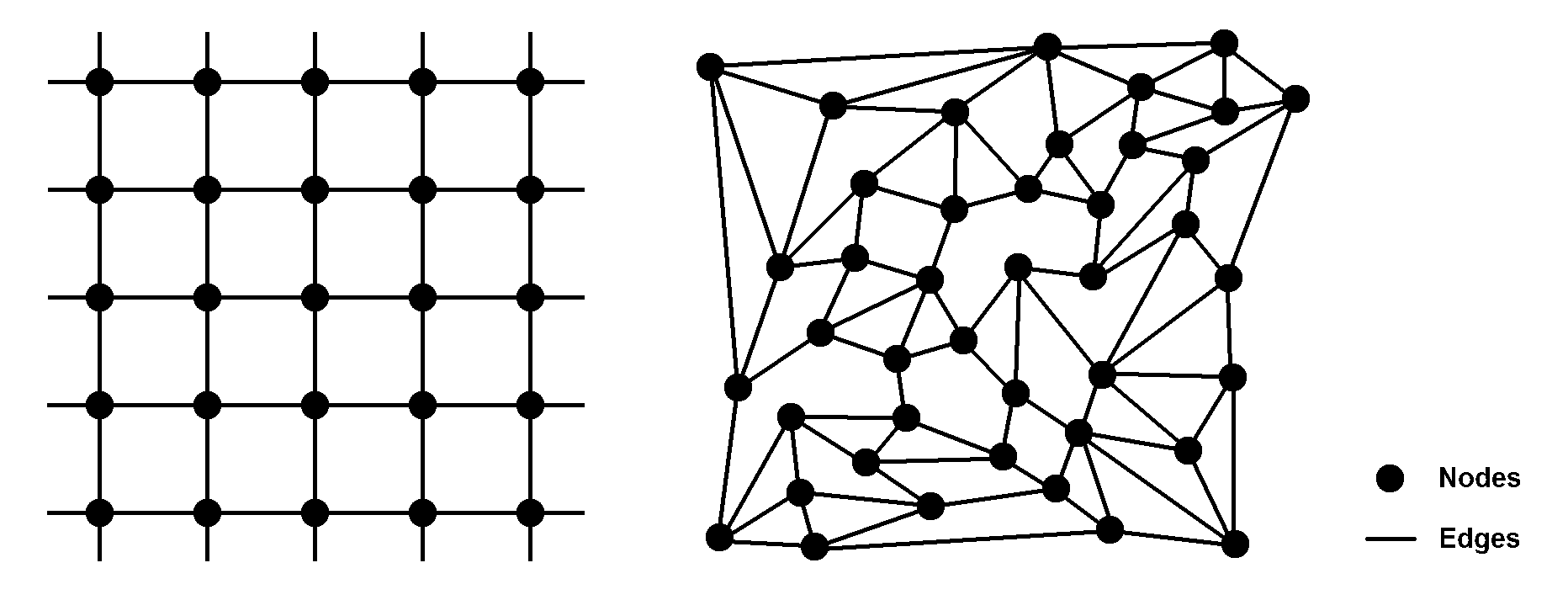}
\caption{Torus (left) and $4$-connected planar graph (right)} 
\label{figgraph}
\end{center}
\end{figure}

\paragraph{Our contribution.}

In this paper, we generalize the result of \cite{Trig} to $4$-connected planar graphs (see Figure~\ref{figgraph}). We show that reliable broadcast can be guaranteed when $D > Z$, $Z$ being the maximal number of edges per polygon. We also show that this bound is tight: if we only have $D \geq Z$, no algorithm can guarantee reliable broadcast for this class of graphs.

Then, if we assume that the delay between two activations of a same process is bounded, we show that reliable broadcast can be achieved in $O(d)$ time units, $d$ being the network diameter. So, tolerating Byzantine failures yields the same time complexity as a simple broadcast.

Finally, we show that, unlike previous solutions \cite{K04c,BV05c,PP05j,CtrZ,Scalbyz,Trig}, the local memory required for broadcasting is $O(M)$ (instead of $O(2^M)$), $M$ being the maximal size of an information message.

\paragraph{Organization of the paper} 
In Section~\ref{sec_setting}, we present the hypotheses and describe the broadcast protocol.
In Section~\ref{sec_relprop}, we prove the condition for reliable broadcast, and show its tightness.
In Section~\ref{sec_timprop}, we establish the time complexity.
Finally, in Section~\ref{sec_memo}, we discuss about the memory requirements.

\section{Setting}

\label{sec_setting}
In this section, we present our hypotheses and describe the broadcast protocol.

\subsection{Hypotheses}

\paragraph{Topology}

Let $\mathcal{G} = (G,E)$ be a graph representing the topology of the network. $G$ denotes the \emph{nodes}, and $E$ denotes the \emph{edges} connecting two nodes. The graph $\mathcal{G}$ is \emph{planar}: there exists a bi-dimensional representation of this graph where edges do not cross. Besides, we assume that the graph is $4$-connected: to disconnect the graph, at least $4$ nodes must be removed (see Definition~\ref{def_connec}). From this hypothesis, each node connects at least $4$ edges.

As the graph is planar, the edges delimit \emph{polygons} (see Figure~\ref{figgraph} and Definition~\ref{def_poly}).
Let $Z \geq 3$ be the maximal number of edges per polygon, and let $Y \geq 4$ be the maximal number of edges per node. 
$Z$ is a parameter of the algorithm.

\paragraph{Network}

Two nodes (or \emph{processes}) connected by an edge (or \emph{channel}) are called \emph{neighbors}. A node can only send messages to its neighbors. Some nodes are $correct$ and follow the protocol described thereafter. The other nodes are \emph{Byzantine}, and have a totally unpredictable behavior. The correct nodes do not know which nodes are Byzantine.

We consider an \emph{asynchronous} network: any message sent is eventually received, but it can be at any time. We assume that, in an infinite execution, any process is activated infinitely often; however, we make no hypothesis on the order of activation of the processes. Finally, we assume \emph{authenticated channels} (or ``oral'' model): each node has a unique identifier, and knows the identifier of its neighbors.
Therefore, when a node receives a message from a neighbor $p$, it knows that $p$ is the actual author of the message.

\subsection{Protocol}

\paragraph{Preliminaries}

An arbitrary correct node, called the \emph{source}, wants to broadcast an information $m_0$ in the network. We say that a correct node \emph{multicasts} a message when it sends it to all its neighbors, and \emph{delivers} $m$ when it permanently considers that $m$ was broadcast by the source. We say that we achieve \emph{reliable broadcast} if all correct nodes eventually deliver $m_0$.

\paragraph{Principle of the protocol}
We use the same underlying principle as in \cite{Trig}: to actually deliver an information message, a node must receive it from a direct neighbor $q$, but also (indirectly) from another node located at at most $Z-2$ hops. The intuitive idea is that, if two Byzantine nodes are distant from more than $Z$ hops, they can never cooperate to make a correct node deliver a false information.

Besides generalizing the aforementioned protocol to planar graphs, our new protocol improves memory efficiency. Indeed, instead of storing all received messages in a set $Rec$, a correct node uses a variable $Rec(q)$ for each neighbor $q$, storing only the last message received from $q$. This modification enables to reduce the memory required by the nodes (see Section~\ref{sec_memo}).

The messages exchanged in the protocol are tuples of the form $(m,S)$, where $m$ is the information broadcast by the source (or pretending to be it), and $S$ is a set containing the identifiers of the nodes already visited by the message.

\paragraph{Description of the protocol}

\begin{itemize}
\item The source multicasts an arbitrary information $m_0$.
\item The correct nodes that are neighbors of the source wait until they receive an information $m$ from the source, then deliver $m$ and multicast $(m,\o)$.
\item The other correct nodes have the following behavior:
\begin{itemize}

\item When $(m,S)$ is received from a neighbor $q$, with $q \notin S$ and $card(S) \leq Z-3$: assign the value $(m,S)$ to $Rec(q)$ and multicast $(m, S \cup \{q\})$.

\item When there exists $m$, $p$, $q$ and $S$ such that $q \neq p$, $q \notin S$, $Rec(q) = (m,\o)$ and $Rec(p) = (m,S)$: deliver $m$, multicast $(m,\o)$ and stop.
\end{itemize}
\end{itemize}

\section{Condition for reliable broadcast}

\label{sec_relprop}

In this section, we prove the main result of the paper: if $D > Z$, we achieve reliable broadcast. We also show that this bound on $D$ is tight: if we only have $D \geq Z$, no algorithm can guarantee reliable broadcast on this class of graphs.

\subsection{Definitions}

\begin{definition}[Path and circular path]
A \emph{path} is a sequence of nodes $(u_1,\dots,u_n)$ such that $u_i$ and $u_{i+1}$ are neighbors. This path is \emph{circular} if $u_1$ and $u_n$ are also neighbors. Unless we mention it, we do not require that these nodes are distinct.
\end{definition}

\begin{definition}[Node-cut and $k$-connected network]
\label{def_connec}
As set $S$ of nodes is a \emph{node-cut} if the graph $G - S$ is disconnected, that is: there exists a pair of nodes $\{p,q\} \notin S$ such that no path connects $p$ and $q$ in $G - S$. The network is \emph{$k$-connected} if no node-cut contains less than $k$ nodes. 
\end{definition}

\begin{definition}[Polygon]
\label{def_poly}
A \emph{polygon} is a circular path that does not surround any node in the bidimensionnal representation of the planar graph.
\end{definition}

\begin{definition}[Neighbor and adjacent polygons]
Two polygons are \emph{neighbors} if they share at least one node, and \emph{adjacent} if they share an edge.
\end{definition}

\begin{definition}[Polygonal path]
A \emph{polygonal path} is a sequence of polygons $(P_1,\dots,P_n)$ such that $P_{i}$ and $P_{i+1}$ are adjacent.
\end{definition}

\begin{definition}[Connected polygons]
A set $S$ of polygons is \emph{connected} if, for each pair of polygons $(P,Q)$ of S,
there exists a polygonal path $(P,P_1,\dots,P_n,Q)$ in $S$.
\end{definition}

\begin{definition}[Correct and Byzantine polygons]
A polygon is \emph{correct} if all its nodes are correct. Otherwise, it is Byzantine.
\end{definition}

\subsection{Main theorem}

Let us show that, if $D > Z$, we achieve reliable broadcast (Theorem~\ref{threl}).

\label{secthrel}

\begin{lemma}
\label{byznei}
Let us suppose that $D > Z$. Then, if two polygons are neighbors, the set of their nodes contains at most one Byzantine node.
\end{lemma}

\begin{proof}
The proof is by contradiction.
Let us suppose the opposite:
there exist two neighbor polygons $P$ and $Q$, and the set of their nodes contains two distinct Byzantine nodes $b_1$ and $b_2$.

As $P$ and $Q$ are neighbors, let $u$ be a node shared by $P$ and $Q$.
Let $(u,p_1,\dots,p_n)$ be a circular path on $P$, and let $(u,q_1,\dots,q_m)$ be a circular path on $Q$.
Therefore, $(u,p_1,\dots,p_n,u,q_1,\dots,q_m)$ is a circular path containing all the nodes of $P$ and $Q$.

As this circular path contains at most $2Z$ hops, two nodes of this path are distant of at most $Z$ hops. In particular, $b_1$ and $b_2$ are distant of at most $Z$ hops, which contradicts $D > Z$. Hence, the result.
\end{proof}

\begin{lemma}
\label{komon}
Let $v$ be a node, and let $V$ be the set of polygons containing $v$.
Then, $v$ is the only node common to these polygons.
\end{lemma}

\begin{proof}
Let us suppose the opposite: the exists a node $w \neq v$ common to these polygons.
Let $P$ be a polygon containing $v$.
Let $q_1$ and $q_2$ be the two neighbors of $v$ contained by $P$.
Let $Q_1$ (resp. $Q_2$) be the polygon adjacent to $P$ containing $v$ and $q_1$ (resp. $q_2$).
Let $S$ be the set of nodes contained by $P$.
As a polygon contains at least $3$ nodes, $S - \{v,w\}$ contains at least one node.
Then, as $w$ is also common to $P$, $Q_1$ and $Q_2$, $\{v,w\}$ is a node-cut isolating $S - \{v,w\}$ from the rest of the network.
This is impossible, as the network is $4$-connected.
Hence, the result.
\end{proof}

\begin{lemma}
\label{lpcorr}
If $D > Z$, each correct node belongs to at least one correct polygon.
\end{lemma}

\begin{proof}
Let us suppose the opposite: there exists a correct node $v$ that does not belong to any correct polygon.
Let $V$ be the set of polygons containing $v$.
Let $P_1$ and $P_2$ be two polygons of $V$.
As $P_1$ and $P_2$ are Byzantine, according to Lemma~\ref{byznei}, they share the same Byzantine node $b$
Therefore, by induction, all the polygons of $V$ share the same Byzantine node $b$.
But according to Lemma~\ref{komon}, $v$ is the only node shared by the polygons of $V$.
Therefore, $b = v$, and $v$ is Byzantine: contradiction. Hence, the result.
\end{proof}

\begin{lemma}
\label{lcircupath}
Let $v$ be a node, and let $V$ be the set of polygons containing $v$. Let $X$ be the set of nodes contained by the polygons of $V$. 
Then, there exists a circular path $(q_1,\dots,q_m)$ such that nodes $\{q_1,\dots,q_m\}$ are distinct and that contains all nodes of $X - \{v\}$, and only contains nodes of $X$.
\end{lemma}

\begin{proof}
Let $(e_1,\dots,e_n)$ be the edges connected to $v$, ordered clockwise, and let $e_{n+1} = e_1$. Let $u_i$ be the node connected to $v$ by $e_i$.
If, $\forall i \in \{1,\dots,n\}$, there exists a polygon containing the edges $e_i$ and $e_{i+1}$, go to paragraph $1$. Else, go to paragraph $2$.

\begin{enumerate}

\item Let $P_i$ be the polygon containing the edges $e_i$ and $e_{i+1}$. Let $(v,u_i,p^i_1,p^i_2,\dots,u_{i+1})$ be a circular path on $P_i$, ordered clockwise.
We define a path $(u_1,p^1_1,p^1_2,\dots,u_2,p^2_1,p^2_2,\dots,u_{n+1}) = (q_1,\dots,q_{m+1})$, containing all the nodes of $X - \{v\}$.
Let us show that the nodes $\{q_1,\dots,q_m\}$ are distinct.
Let us suppose the opposite: there exists $k$ and $k' > k$ such that $u_k = u_{k'}$. Then, $\{u_k,v\}$ is a node-cut disconnecting $\{u_{k+1},\dots,u_{k'-1}\}$ from the rest of the network, which is impossible as the network is $4$-connected. Thus, the nodes are distinct. Hence, the result.

\item 
Let $k$ be the first integer such that $e_{k}$ and $e_{k+1}$ do not belong to any polygon. Let us notice that there is no other integer $k' > k$ satisfying this property -- otherwise, $\{v\}$ would be a node-cut isolating $u_k$ from $u_{k'}$.
Then, let $(e'_1,\dots,e'_n)$ be the edges connected to $v$, clockwise, such that $e'_1 = e_{k+1}$.

Let $P_i$ be the polygon containing the edges $e'_i$ and $e'_{i+1}$. Let $(v,u_i,p^i_1,p^i_2,\dots,u_{i+1})$ be a circular path on $P_i$, ordered clockwise.
We define a path $(u_1,p^1_1,p^1_2,\dots,u_2,p^2_1,p^2_2,\dots,u_{n}) = (q_1,\dots,q_{m-1})$, containing all nodes of $X - \{v\}$.
For the same reasons as in paragraph $1$, the nodes $\{q_1,\dots,q_{m-1}\}$ are distinct.
Hence, the result, if we take $q_m = v$.

\end{enumerate}

\end{proof}

\begin{lemma}
\label{lcircle}
Let $v$ be a node, and let $V$ be the set of polygons containing $v$. Let $S$ be the set of polygons that are not is $V$, but are neighbors with a polygon of $V$. Then, $S$ is connected.
\end{lemma}

\begin{proof}
Let $(q_1,\dots,q_m)$ be the circular path of Lemma~\ref{lcircupath}.
Then, $S = S_1 \cup \dots \cup S_m$, where $S_i$ is the set of polygons containing $q_i$.
If each set $S_i$ is connected, as $S_i$ and $S_{i+1}$ share a polygon containing $q_i$ and $q_{i+1}$, $S$ is connected. Now, let us suppose that there exists a $k$ such that $S_k$ is not connected.

$S_k$ contains only two disconnected parts, otherwise $\{v\}$ would be a node-cut.
Let $(q'_1,\dots,q'_m)$ be a circular path containing nodes $\{q_1,\dots,q_m\}$, ordered clockwise, such that $q'_1 = q_k$.
Let $S'_1$ (resp. $S'_{m+1}$) be the part of $S_k$ containing the node $q_2$ (resp. $q_m$).
$\forall i \in \{2,\dots,m\}$, let $S'i$ be the set of polygons containing $q'i$. Then, $S = S_1 \cup \dots \cup S_m= S'1 \cup \dots \cup S'_{m+1}$.
Let us prove the following property $\mathcal{P}_i$ by induction, $\forall i \in \{1,\dots,m+1\}$: $S'_1 \cup \dots \cup S'_i$ is connected.

\begin{itemize}
\item $\mathcal{P}_1$ is true, as $S'_1$ is connected.
\item Let us suppose that $\mathcal{P}_i$ is true, for $i \in \{1,\dots,m\}$. Let us suppose that $S'_1 \cup \dots \cup S'_{i+1}$ is not connected.
It implies that $S'_{i+1}$ is not connected.
$S'_{i+1}$ contains only two disconnected parts, otherwise $\{q'_{i+1}\}$ would be a node-cut.
Let $S'^A_{i+1}$ be the part containing the node $q'_i$, and let $S'^B_{i+1}$ be the other part.
Then, $\{q'_1,v,q'_{i+1}\}$ is a node-cut isolating $S'_1 \cup \dots S'_i \cup S'^A_{i+1}$ from $S'^B_{i+1}$, which is impossible as the network is $4$-connected.
Thus, $\mathcal{P}_{i+1}$ is true.
\end{itemize}

Therefore, $\mathcal{P}_{m+1}$ is true, and $S$ is connected.
\end{proof}

\begin{lemma}
\label{byzcircle}
Let us suppose that $D > Z$.
Let $(P,P_1,\dots,P_n,Q)$ be a polygonal path such that $P$ and $Q$ are correct, and $\{P_1,\dots,P_n\}$ are Byzantine. Then, there exists a polygonal path $(P,Q_1,\dots,Q_m,Q)$ such that $\{Q_1,\dots,Q_m\}$ are correct.
\end{lemma}

\begin{proof}
According to Lemma~\ref{byznei}, $P_i$ and $P_{i+1}$ share the same Byzantine node $b$.
Therefore, by induction, the polygons $\{P_1,\dots,P_n\}$ share the same Byzantine node $b$.

Let $V$ be the set of polygons containing $b$, and let $S$ be the set of polygons that are not in $V$, but are neighbors to a polygon of $V$. 
As $V$ contains $P_1$ and $P_n$, by definition, $S$ contains $P$ and $Q$.
According to Lemma~\ref{lcircle}, $S$ is connected: there exists a polygonal path $(P,Q_1,\dots,Q_m,Q)$ in $S$.
To complete the proof, let us show that the polygons of $S$ are correct.

Let us suppose the opposite: there exists a polygon $P'$ of $S$ that is Byzantine.
Let $b'$ be the Byzantine node contained by $P'$.
Then, as $P'$ has a neighbor polygon in $V$, according to Lemma~\ref{byznei}, $b' = b$. It implies that $P'$ belongs to $V$: contradiction. Thus, the polygons of $S$ are correct. Hence, the result.

\end{proof}

\begin{lemma}
\label{lppath}
If $D > Z$, the set of correct polygons is connected.
\end{lemma}

\begin{proof}
Let $P$ and $Q$ be two correct polygons, and let $(P,P_1,\dots,P_n,Q)$ be a polygonal path.
If $\{P_1,\dots,P_n\}$ are correct, the result is trivial. Otherwise, let us consider the following process.

Let $N$ be the smallest integer such that $P_N$ is Byzantine, and let $M$ be the smallest integer greater than $N$ such that $P_{M+1}$ is correct. Then, according to Lemma~\ref{byzcircle}, there exists a polygonal path $(P_{N-1},Q_1,\dots,Q_m,P_{M+1})$ such that the polygons $\{Q_1,\dots,Q_m\}$ are correct. Therefore, we can replace the sequence $(P_N,\dots,P_M)$ by $(Q_1,\dots,Q_m)$. We repeat this process until all the polygons of the path are correct.
\end{proof}

\begin{lemma}
\label{lsafe}
Let us suppose that $D \geq Z$. Then, if a correct node delivers an information, it is necessarily $m_0$.
\end{lemma}

\begin{proof}
The proof is by contradiction. Let us suppose the opposite: $D \geq Z$, yet at least one correct node delivers $m' \neq m_0$.
Let $u$ be the first correct node to deliver $m'$. It implies that there exists $p$, $q$ and $S$ such that $q \neq p$, $q \notin S$, $Rec(q) = (m',\o)$ and $Rec(p) = (m',S)$.

$Rec(q) = (m',\o)$ implies that $u$ received $(m',\o)$ from a neighbor $q$. Let us suppose that $q$ is correct. Then, as $q$ sent $(m',\o)$, it implies that $q$ delivered $m'$. This is impossible, as $u$ is the first correct node to deliver $m'$. So $q$ is necessarily Byzantine.
Besides, according to the protocol, $Rec(p) = (m',S)$ implies that $card(S) \leq Z-3$.

Let us prove the following property $\mathcal{P}_i$ by induction,
for $0 \leq i \leq card(S)$: a correct node $p_i$, located at $i + 2$ hops or less from $q$, sent $(m',S_i)$ with $card(S_i) = card(S) - i$.

\begin{itemize}
\item First, let us show that $\mathcal{P}_0$ is true.
$Rec(p) = (m',S)$ implies that $p$ sent $(m',S)$.
Let us suppose that $p$ is Byzantine. Then, as $q$ is also Byzantine, $D \leq 2$, which is impossible as $D \geq Z \geq 3$. So $p$ is necessarily correct, and $\mathcal{P}_0$ is true if we take $p_0 = p$ and $S_0 = S$.
If $Z = 3$, ignore the following step.

\item Let us suppose that $\mathcal{P}_i$ is true, with $i < card(S)$.
As $card(S_i) = card(S) - i \geq 1$, $p_i$ necessarily received $(m',S_{i+1})$ from a node $p_{i+1}$ located at $i+3$ hops or less from $q$, with  $S_i = S_{i+1} \cup \{p_{i+1}\}$ and $p_{i+1} \notin S_{i+1}$.
Thus, we have $card(S_{i+1}) = card(S_i) - 1 = card(S) - i - 1$.
Let us suppose that $p_{i+1}$ is Byzantine.
Then, as $q$ is also Byzantine, $D \leq i+3 \leq card(S)+2 < Z$, which is impossible as $D \geq Z$. So $p_{i+1}$ is necessarily correct, and $\mathcal{P}_{i+1}$ is true.

\end{itemize}

Therefore, $\mathcal{P}_{card(S)}$ is true, and $p_{card(S)}$ sent $(m',\o)$,
as $card(S_{card(S)}) = card(S) - card(S) = 0$. According to the protocol, it implies that $p_{card(S)}$ delivered $m'$ before $u$, which contradicts our initial hypothesis. Hence, the result. 
\end{proof}

\begin{lemma}
\label{lline}
Let us suppose that $D \geq Z$. Let $(u_1,\dots,u_n)$ be a path of distinct correct nodes,
with $3 \leq n \leq Z$, such that $u_1$ and $u_n$ deliver $m_0$. Then, at least one of the nodes $u_2$ and $u_{n-1}$ delivers $m_0$.

\end{lemma}

\begin{proof}
As $u_1$ and $u_n$ deliver $m_0$, and therefore multicast $(m_0,\o)$, let $E_1$ and $E_2$ be the two following events: ($E_1$) $u_2$ receives $(m_0,\o)$ from  $u_1$ and ($E_2$) $u_{n-1}$ receives $(m_0,\o)$ from $u_n$. Let us suppose that $E_2$ is the first event to occur.
As $u_n$ delivers $m_0$, according to the protocol, $u_n$ stops.
Therefore, for the node $u_{n-1}$, $Rec(u_n) = (m_0,\o)$ until the end of the execution.

Let us prove the following property $\mathcal{P}_i$ by induction, for $1 \leq i \leq n-2$:
$u_i$ multicasts $(m_0,S_i)$, with $S_i \subseteq \{u_1,\dots,u_{n-2}\}$ and $card(S_i) \leq i-1$.
\begin{itemize}

\item As $u_1$ delivers $m_0$, $u_1$ multicasts $(m_0,\o)$. Therefore, $\mathcal{P}_1$ is true if we take $S_0 = \o$

\item Let us suppose that $\mathcal{P}_i$ is true, for $i < n-2$.
Then, $u_{i+1}$ receives $(m_0,S_i)$ from $u_i$, with $card(S_i) \leq i-1 < n-3 \leq Z-3$.
When it does, two possibilities:
\begin{itemize}
\item If $u_{i+1}$ has stopped, $u_{i+1}$ has necessarily delivered an information. As $D \geq Z$, according to Lemma~\ref{lsafe}, this information was $m_0$. Thus, according to the protocol, $u_{i+1}$ has already multicast $(m_0,\o)$, and $\mathcal{P}_{i+1}$ is true if we take $S_{i+1} = \o$.
\item Otherwise, as $card(S_i) \leq Z-3$, $u_{i+1}$ multicasts $(m_0,S_i \cup \{u_i\})$. Thus, $\mathcal{P}_{i+1}$ is true if we take $S_{i+1} = S_i \cup \{u_i\}$.
\end{itemize}
\end{itemize}

Therefore, $\mathcal{P}_{n-2}$ is true, and $u_{n-1}$ receives $(m_0,S_{n-2})$ from $u_{n-2}$,
with $S_{n-2} \subseteq \{u_1,\dots,u_{n-2}\}$ and $card(S_{n-2}) \leq n-3 \leq Z-3$.
Thus, for the node $u_{n-1}$, $Rec(u_{n-2}) = (m_0,S_{n-2})$, with $u_n \notin S_{n-2}$.
Thus, as we already have $Rec(u_n) = (m_0,\o)$, according to the protocol, $u_{n-1}$ delivers $m_0$.

If $E_1$ is the first event to occur, by a perfectly symmetric reasoning, we show that $u_2$ delivers $m_0$. Hence, the result.
\end{proof}

\begin{lemma}
\label{lxircle}
Let us suppose that $D \geq Z$.
Let $P$ be a correct polygon, and let $p_1$ and $p_2$ be two neighbor nodes of $P$ that deliver $m_0$. Then, all the nodes of $P$ deliver $m_0$.
\end{lemma}

\begin{proof}
Let $z \leq Z$ be the number of nodes of $P$.
Let us prove the following property $\mathcal{P}_i$ by induction, for $1 \leq i \leq z-1$:
there exists a path of $i+1$ nodes of $P$ that deliver $m_0$.

\begin{itemize}
\item $\mathcal{P}_1$ is true, as $(p_1,p_2)$ is a path of $2$ nodes that deliver $m_0$.

\item Let us suppose that $\mathcal{P}_i$ is true for $i < z-1$.
Let $(u_1,\dots,u_{i+1})$ be a path of $i+1$ nodes that deliver $m_0$.
Let $\{q_1,\dots,q_n\}$ be $n$ nodes such that $(u_1,\dots,u_{i+1},q_1,\dots,q_n,u_1)$ is a circular path on $P$.
Then, $(u_{i+1},q_1,\dots,q_n,u_1)$ is a path of correct nodes where 
$u_{i+1}$ and $u_1$ deliver $m_0$.
Therefore, according to Lemma~\ref{lline}, at least one of the nodes $q_1$ and $q_n$ deliver $m_0$.
Thus, at least one of the paths $(q_n,u_1,\dots,u_{i+1})$ and $(u_1,\dots,u_{i+1},q_1)$ contains $i+2$ nodes of $P$
that deliver $m_0$, and $\mathcal{P}_{i+1}$ is true.

\end{itemize}

Therefore, $\mathcal{P}_{z-1}$ is true, and the $z$ nodes of $P$ deliver $m_0$.
\end{proof}

\begin{theorem}
\label{threl}
If $D > Z$, we achieve reliable broadcast.
\end{theorem}

\begin{proof}
Let $s$ be the source and let $p$ be a correct node.
According to Lemma~\ref{lpcorr}, $s$ belongs to a correct polygon $P$ and $p$ belongs to a correct polygon $P'$.
According to Lemma~\ref{lppath}, there exists a correct polygonal path $(Q_1,\dots,Q_n)$ such that
$Q_1 = P$ and $Q_n = P'$.

Let us prove the following property $\mathcal{P}_i$ by induction, for
$1 \leq i \leq n$:
all the nodes of $Q_i$ deliver $m_0$.

\begin{itemize}
\item First, let us show that $\mathcal{P}_1$ is true. Let $q$ be a neighbor of $s$ on $Q_1$.
As $Q_1$ is correct, according to the protocol, $q$ delivers $m_0$.
Then, according to Lemma~\ref{lxircle}, $\mathcal{P}_1$ is true.

\item Let us suppose that $\mathcal{P}_i$ is true, for $i < n$.
Let $u_1$ and $u_2$ be the two nodes shared by $Q_i$ and $Q_{i+1}$.
As $\mathcal{P}_i$ is true, $u_1$ and $u_2$ deliver $m_0$.
Then, according to Lemma~\ref{lxircle}, $\mathcal{P}_{i+1}$ is true.
\end{itemize}

Thus, $\mathcal{P}_n$ is true, and $p$ delivers $m_0$. Hence, the result.

\end{proof}

\subsection{Bounds tightness}

Let us show that the bound on $D$ (Theorem~\ref{threl}) cannot be improved.

\label{secthtight}

\begin{theorem}
\label{thtight}
If $D \geq Z$, no algorithm can guarantee reliable broadcast on $4$-connected planar graphs.
\end{theorem}

\begin{proof}
Let us suppose the opposite: there exists an algorithm guaranteeing reliable broadcast on $4$-connected planar graphs for $D \geq Z$.
Let us consider the network of Figure~\ref{figcrit}.

\begin{figure}[H]
\begin{center}
\includegraphics[width=10cm]{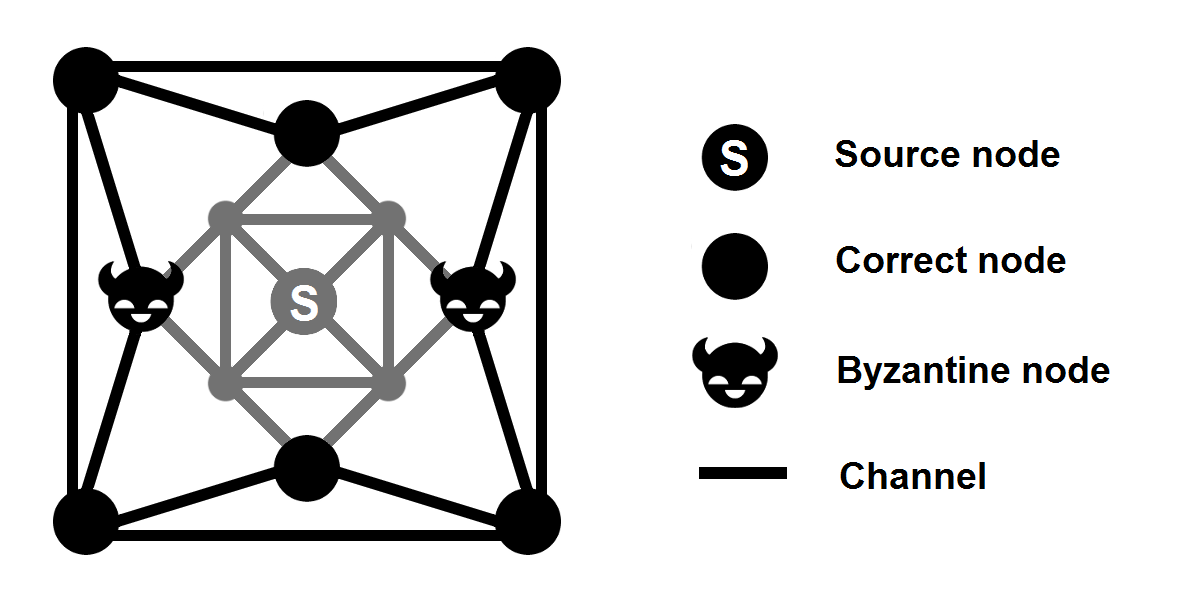}
\caption{Critical case for $D \geq Z$} 
\label{figcrit}
\end{center}
\end{figure}

In this network, $D = Z = 4$, thus $D \geq Z$ is satisfied. Here, we have $4$ nodes ($2$ correct, $2$ Byzantine) forming a node-cut that isolates the grey part of the network, which contains the source.

As there is a perfect symmetry between the $2$ correct nodes and the $2$ Byzantine nodes, the outer nodes can never determine $m_0$ with certitude, and reliable broadcast is impossible. This contradiction achieves the proof.
\end{proof}

Nevertheless, notice that it does not make the condition $D > Z$ necessary for all graphs: the necessary and sufficient condition to achieve byzantine resilient broadcast may be more complex than the distance between Byzantine failures. We leave this as an open question.

\section{Time complexity}

\label{sec_timprop}

In this section, we assume that the delay between two activations of the same process has an upper bound $T$. Then, we show that reliable broadcast is achieved in $O(d)$ time units, $d$ being the diameter of the network. This is the same time complexity as a simple broadcast, where any information received is retransmitted without verification.

\begin{lemma}
\label{lsizepath}
Let $p$ be a node located a $L \geq 1$ hops from the source.
Then, there exists a correct polygonal path of at most $Y^{3}ZL$ polygons connecting $p$ to the source.
\end{lemma}

\begin{proof}

Let $P$ be a correct polygon containing the source $s$, and let $P'$ be a correct polygon containing $p$. Such polygons exist, according to Lemma~\ref{lpcorr}.
Let $(u_1,\dots,u_{L+1})$ be a path connecting $s$ and $p$, and let $U_i$ be the set of polygons containing $u_i$. Each set $U_i$ is connected, otherwise $\{u_i\}$ would be a node-cut.
Therefore, $U = U_1 \cup \dots \cup U_{L+1}$ is connected. As each set $U_i$ contains at most $Y$ polygons, $U$ contains at most $Y(L+1)$ polygons.

Therefore, there exists a polygonal path $(P_1,\dots,P_n)$ of at most Y(L+1) polygons, with $P_1 = P$ and $P_n = P'$.
If this path is correct, the result is trivial. Otherwise,
let $(P_N,\dots,P_M)$ be a sequence of Byzantine nodes, as defined
in Lemma~\ref{lppath}.

Let us consider the proof of Lemma~\ref{lcircle}.
The circular path $(q_1,\dots,q_m)$ contains at most $YZ$ nodes,
and each set $S_i$ contains at most $Y$ polygons.
Thus, the set $S = S_1 \cup \dots \cup S_m$ contains at most $Y^{2}Z$ polygons.

Therefore, according to the proof of Lemma~\ref{lppath}, $(P_N,\dots,P_M)$ can be replaced by a sequence of at most $Y^{2}Z$ polygons. As the number of Byzantine sequences in $(P_1,\dots,P_n)$ is strictly inferior to $n/2$,
the correct path thus obtained contains at most
$Y^{2}Zn/2 \leq Y^{2}ZY(L+1)/2 \leq Y^{3}ZL$ polygons.
\end{proof}

\begin{theorem}
\label{thtime}
Reliable broadcast is achieved in $O(d)$ time units.
\end{theorem}

\begin{proof}
Let us suppose that the source broadcasts $m_0$ at a date $t_0$.

Let $Q$ be a correct polygon, and let us suppose that two nodes of $Q$ have delivered $m_0$ at a date $t$.
Then, according to the proof of Lemma~\ref{lxircle}, a third node delivers $m_0$ before $t + ZT$, and so forth.
Thus, all nodes of $Q$ deliver $m_0$ before $t + Z^{2}T$.
Similarly, all nodes of $P$ deliver $m_0$ before $t_0 + Z^{2}T$.

According to Lemma~\ref{lsizepath}, 
for any node $p$ located at $L \geq 1$ hops from the source,
there exists a correct polygonal path
of $Y^{3}ZL$ polygons connecting this node to the source.
Thus, according to the proof of Theorem~\ref{threl},
$p$ delivers $m_0$ before $t_0 + Y^{3}Z^{3}LT$.

Therefore, as $L \leq d$, reliable broadcast is achieved in $Y^{3}Z^{3}Td$ time units. Thus, as $Y$, $Z$ and $T$ are bounded, reliable broadcast is achieved in a $O(d)$ time.

\end{proof}

\section{Required memory}

\label{sec_memo}
\label{exsafe}

In this section, we show that our solution is the first Byzantine resilient broadcast in sparse multi-hop networks where the used memory increases linearly with the size of informations, and not exponentially.

Indeed, the existing solutions \cite{K04c, BV05c, NT09j, CtrZ, Trig, Scalbyz}, the nodes are supposed to store as many information messages $m$ as necessary.
However, the Byzantine nodes can potentially broadcast all possible false informations $m' \neq m_0$. This strategy is referred to as \emph{exhaustion} in the literature \cite{veri1, veri2}.
Therefore, the correct nodes implicitly require $O(2^M)$ bits of memory to ensure reliable broadcast, $M$ being the maximal number of bits of an information $m$.

In our protocol, we made the following modification : instead of storing all the messages received, we only store the last message received from a neighbor $q$ in the variable $Rec(q)$. Thus, the nodes only require $O(M)$ bits of memory. More precisely, let us consider a finite network, and let $X$ be the maximal number of bits of a node identifier. As the largest tuple $(m,S)$ that a correct node can accept verifies $card(S) \leq Z$, each
variable $Rec$ requires at most $M + ZX$ bits. Thus, each correct node requires at most $Y(M + ZX)$ bits of memory.

Concerning the memory required in channels, the problem is the same for all solutions: we must assume that the delay between two activations of a same process belongs to an interval $\left[ T_1, T_2 \right]$, $T_1 > 0$ -- otherwise, the memory is impossible to bound. Indeed, let $N$ be the smallest integer such that $N > T_2/T_1$.
Then, as a node receives all the messages of its channels when activated, a channel connecting two correct nodes contains at most $N$ tuples $(m,S)$.
Besides, if a channel is connected to a Byzantine node, it can be overflowed without consequences: it is unimportant that a Byzantine node receives messages, and the messages received from a Byzantine node are already unpredictable. Thus, each channel requires at most $N(M + XZ)$ bits of memory.

Therefore, the local memory required in now $O(M)$ instead of $O(2^M)$.

\section{Conclusion}

We generalized the condition on the distance between Byzantine nodes to a class of planar graphs, and shown its tightness. Our solution has the same time complexity as a basic broadcast, and requires less memory than the previous solutions.

An open problem is to find more involved criteria for the placement of Byzantine failures, and to extend it to more general graphs. Also, even if we already have a linear time complexity, some optimizations could be made to further reduce the time to deliver genuine information.

\bibliographystyle{plain}
\bibliography{biblio}

\end{document}